\documentclass[runningheads]{llncs}
\usepackage{graphicx,amssymb,amsmath,color}
\usepackage{tabularx,lipsum,environ}
\makeatletter
\newcommand{\problemtitle}[1]{\gdef\@problemtitle{#1}}
\newcommand{\probleminput}[1]{\gdef\@probleminput{#1}}
\newcommand{\problemquestion}[1]{\gdef\@problemquestion{#1}}
\newtheorem{fact}{Fact}
\NewEnviron{myproblem}{
  \problemtitle{}\probleminput{}\problemquestion{}
  \BODY
  \par\addvspace{.5\baselineskip}
  \noindent
  \begin{tabularx}{\textwidth}{@{\hspace{\parindent}} l X c}
    \multicolumn{2}{@{\hspace{\parindent}}l}{\textsc{\@problemtitle}} \\
    \textbf{Input:} & \@probleminput \\
    \textbf{Output:} & \@problemquestion
  \end{tabularx}
  \par\addvspace{.5\baselineskip}
}

\def\s{{\bf s}}
\def\w{{\bf w}}
\newcommand{\remove}[1]{}

\begin{document}
\title{On the complexity and approximability of Bounded access Lempel Ziv coding\thanks{Ferdinando Cicalese is member of the Gruppo Nazionale Calcolo Scientifico-Istituto Nazionale di Alta Matematica (GNCS-INdAM).}}
\titlerunning{Hardness of Bounded access Lempel Ziv coding}
%
\author{Ferdinando Cicalese 
\and
Francesca Ugazio 
}
\authorrunning{F. Cicalese and F. Ugazio}
\institute{Department of Computer Science, University of Verona, Italy \\
\email{ferdinando.cicalese@univr.it,francesca.ugazio@studenti.univr.it}}
\maketitle              
\begin{abstract}
We study the complexity of constructing an optimal parsing $\varphi$ of a string $\s = s_1 \dots s_n$ 
under the constraint that given a position $p$ in the original text, 
and the LZ76-like (Lempel Ziv 76) encoding of $T$ based on $\varphi$, 
it is possible to identify/decompress the character $s_p$  by performing at most $c$ accesses to the LZ encoding, for a given integer $c.$
We refer to such a parsing $\varphi$ as a $c$-bounded access LZ parsing or $c$-BLZ parsing of $\s.$  
We show that for any constant $c$ the problem of computing the optimal $c$-BLZ parsing of a string, i.e., the one  with the minimum number of phrases, 
is NP-hard and also APX hard, i.e., no PTAS can exist under the standard complexity assumption $P \neq NP.$ 
We also study the ratio between the sizes of an optimal $c$-BLZ  parsing of a string $\s$ and an optimal LZ76 parsing of $\s$ 
(which can be greedily computed in polynomial time).

\end{abstract}
\thispagestyle{empty}
\newpage

\setcounter{page}{1}
\section{Introduction}
\label{sec:introduction}
LZ76 \cite{LZ76} parses a text $\s$ from left to right by defining each new phrase as the longest substring already appeared in the text plus a new character. 
It is known that such a parsing is the one with the minimum number of phrases among all parsing where each phrase must be chosen to be either a single character or a previously occurred substring plus ad additional character. 
LZ76 parsing is the basis of the LZ77 encoding, where each phrase $\phi$ is encoded as  either (i) the triple $(p,\ell, c)$ where $p$ and $\ell$ are the position and length of the previously occurring substring equal to the longest proper prefix of $\phi$ and $c$ is the last character of $\phi$; or (ii) as the triple $(0,0,c)$ when $\phi$ has exactly one character. Remarkably, LZ76 parsing and LZ77 encoding of a string $\s$ can be done in linear time as well as the decompression of an encoding can also be done in 
linear time, reconstructing the original string from left to right \cite{Rodeh-etal-JACM81}.

However, when  an  encoding is employed for compressed data structures, additional properties are desirable. 
A property that has recently gained attention \cite{LMN24,bannai_et_al:LIPIcs.CPM.2023.3,kempa2017lz,Spire21,Spire23} is the following: given the encoding of a string $\s$ and a position $p$ (or positions $p,q$) we would like to be able to 
return the character $s_p$ (or the substring  $\s[p,q]$) without having to decompress the whole string, and possibly in time $O(1)$ (resp.\ $O(q-p)$), i.e. 
 proportional to the length of the extracted substring and not of the encoding or the text. 

To clarify this issue, let us consider the basic LZ77 encoding of a string $\s$. If we want to extract the character $s_p$ we can  compute the index $i_1$ of the phrase $\varphi_{i_1}$ that contains the $p$th character. However, only if  $s_p$ is the last character of $\varphi_{i_1}$ we would find it directly in the encoding. 
Otherwise, if $s_p$ is in the proper prefix of $\varphi_{i_1},$ from the encoding, we can compute the index $i_2$ of the phrase $\varphi_{i_2}$ and the position $q$ such that  $s_q = s_p$ is part of $\varphi_{i_2}.$ Again, if $s_q$ is the ending character of $\varphi_{i_2}$ we have access to it, otherwise additional {\em hops} will be necessary until we reach a phrase $\varphi_{i_t}$ such that $s_p$ is its end character. 
In such case we say that position $p$ requires $t$ hops or access to extract $s_p$. 

It should be clear that given a parsing $\varphi$ (or equivalently the associated encoding) 
of a string $\s$ we can compute for each position $p \in [|\s|]$ 
the value $hop^{\varphi}(p)$ corresponding to the number of hops necessary to identify $s_p.$ 
We are interested in 
parsings/encodings that guarantees $hop_{\max}^{\varphi}(\s) = \max_i hop^{\varphi}(i)$ to be small.
The problem is that in general for the  LZ76 greedy parsing/encoding $\varphi^* = \varphi_{LZ76}$ the value $hop_{\max}^{\varphi^*}(\s)$ appears to be 
large \cite{LMN24}.

In this paper we study the complexity of the following problem, which was originally considered in \cite{Kreft-Navarro-DCC10} and recently experimentally analyzed in \cite{LMN24,bannai2024height}:

\begin{myproblem}
  \problemtitle{Bounded Access Lempel Ziv Parsing (BLZ)}
  \probleminput{A strings $\s$ and an integer $c$.}
  \problemquestion{A parsing $\varphi^*_c$ of $\s$ with the minimum number of phrases among all LZ-parsings $\varphi_c$  
  of $\s$ satisfying $hop_{\max}^{\varphi_c}(\s) \leq c$}
\end{myproblem}

\noindent
{\bf Our results.} We show that the above problem is NP-hard for any constant $c.$ Moreover, we show that, in terms of approximation,
 the 
problem is $APX$-hard, which, under the standard complexity assumption $NP \neq P$ implies that no PTAS is expected to exist.
Finally we give some results regarding the ratio between the size of an optimal LZ76 parsing and the size of an optimal $c$-BLZ parsing (i.e., guaranteeing $hop_{\max}^{\varphi_c}(\s) \leq c$) for the same string $\s$.

\noindent
{\bf Related work.} The {\sc BLZ} problem was originally introduced in 
\cite{Kreft-Navarro-DCC10} where it was called {\sc LZ-Cost}. In the same paper the authors introduced and focused on another variant of the LZ76 algorithm, called {\sc LZ-End}, where the constraint is that the longest proper prefix of any phrase must have a previous occurrence whose last character is the last character of another phrase.
Both {\sc LZ-End} and {\sc LZ-Cost} were motivated by the possibility of providing fast access to substrings from the compressed text.
 In \cite{Kreft-Navarro-DCC10}, the authors stated not to have found any efficient parsing algorithm for {\sc LZ-Cost}. In \cite{LMN24} several algorithms for constructing $c$-BLZ parsing are presented with extensive empirical evaluation with respect to an optimal LZ76 greedy parsing. The authors show that, with some tolerable cost in terms of additional space, their parsing algorithms can guarantee much better access cost with respect to the original LZ76 parsing. To the best of our knowledge, ours is the first theoretical result on the problem. 
Extensive literature exists on LZ-End parsing. 
Finding an optimal LZ-End parsing was proved NP-hard in \cite{bannai_et_al:LIPIcs.CPM.2023.3}. Previously, the focus had been on comparing the greedy LZ-End parsing to the greedy LZ76 parsing. In \cite{Kempa-Saha-SODA} it is shown that the optimal LZ-End parsing is bounded by 
$O(z \log^2n)$ where $z$ is the size of the LZ76 parsing. The upper bound was later improved by a factor of $\log \log (n/z)$ in \cite{Spire23}. 

\section{Notation and Basic Facts}

{\bf Definitions for strings.}
We denote by $\Sigma^*$ the set of all strings of finite length over a finite alphabet $\Sigma$.
For a string $s \in \Sigma^*$ we denote by $|s|$ the length of $s,$ i.e., the number of its characters.
We refer to the unique string of length $0$ as the null-string, and denote it by $\epsilon.$ We assume that $\epsilon \in \Sigma^*.$ 
For any integer $n \geq 0,$ let $\Sigma^n = \{s \in \Sigma^* \mid |s| = n\}.$
For a string $\s \in \Sigma^n$ and any integers $1 \leq i \leq j \leq n,$ we denote by $s_i$ the $i$th character of $\s$ and by $\s[i,j]$ the sequence of characters $s_i s_{i+1} \dots s_j.$ We say that $s[i,j]$ is a substring of $\s.$ If $i = 1$ (resp.\ $j=n$) we say that $\s[i,j]$ is a {\em prefix} (resp.\ {\em suffix}) of $\s.$ 
A square of a string $\s$ is a substring ${\bf t}$ of $\s$ which occurs twice consecutively in $\s,$ i.e., $\s = \s' {\bf t} {\bf t} \s'',$
for some (possibly empty) strings $\s', \s''.$
If no square is present in $\s$ we say that $\s$ is {\em square-free}.

\begin{definition}[Parsing] A parsing $\varphi$ of a string $\s$ is a partition of $\s$ into consecutive substrings,
$\s[1, i_1] \s[i_1+1, i_2] \cdots \s[i_{t-1}+1, i_t].$ For each $j = 1, \dots, t,$ the substring $\s[i_{j-1}+1, i_{j+1}]$ is called a {\em phrase} of the parsing $\varphi$ and will be denoted by $\varphi_j.$ 
The {\em size of the parsing} $\varphi = \mid \varphi_1 \mid \dots \mid  \varphi_t \mid,$ is the number $t$ of its phrases, 
and will be denoted by $|\varphi (\s)|$ (or also $|\varphi|$ whenever $\s$ is clear from the context).

For a substring $\w$ of $\s,$ we use $\varphi(\w)$ to denote the parsing
induced by $\varphi$ on $\w$ and with $|\varphi(\w)|$ the number of its phrases\footnote{We will generally consider situations where 
$\w$ starts and ends at the boundary of some phrase of $\varphi,$ i.e., $\w = \s[i_j+1, i_{j'}].$ In such a case, we have
$\varphi(\w) = \mid \varphi_{j+1} \mid \varphi_{j+2} \mid \cdots \mid \varphi_{j'} \mid$ and $|\varphi(\w)| = j'-j.$}. 
\end{definition}

\begin{definition}[LZ-parsing and encoding]  
We say that $\varphi$ is an LZ-parsing if for each $j=1, \dots, t$  the phrase $\varphi_j = \s[i_{j-1}+1,i_{j}]$ is either a single character or 
$\s[i_{j-1}+1, i_{j}-1]$ is a substring of $\s$ occurring in $\s$ strictly before the starting position $i_{j-1}+1$ of the the phrase $\varphi_j,$ i.e., 
there exists $\ell \leq i_{j-1},$ such that $\s[\ell, \ell+(i_{j}-i_{j-1})-2] = \s[i_{j-1}+1, i_{j}-1].$
We refer to $\s[\ell, \ell+(i_{j}-i_{j-1})-2] = \s[i_{j-1}+1, i_{j}-1]$ as the 
{\em source}  of $\varphi_j.$ Analogously, for each position 
$p = i_{j-1}+t,$ for $t = 1, \dots, i_{j}-2-i_{j-1}$ we say that position $q = \ell+t$
is the source of $p$ in $\varphi.$

If $\ell+(i_{j}-i_{j-1})-2 \geq i_{j-1}+1$ i.e., the substrings $\s[\ell, \ell+(i_{j}-i_{j-1})-2]$ and $\varphi_j$ overlap, it is easy to see that $\varphi_j$ must be a power of $\s[\ell, i_{j-1}]|,$ i.e., $\varphi_j = \s[\ell, i_{j-1}]\s[\ell, i_{j-1}]\dots
\s[\ell, i_{j-1}]\s[\ell, \ell - 1 +(i_{j}-i_{j-1}) \mod (i_{j-1}-\ell+1)].$
In this case, for each $t = 1, \dots, i_{j}-i_{j-1} - 2$ the source of $p=i_{j-1}+t$
is the position $q = \ell + \left( (t-1) \mod (i_{j-1} - \ell +2) \right),$ i.e., the position of the corresponding character in the repeated substring $\s[\ell, i_{j-1}]|.$

For an $LZ$-parsing we  assume that for any non-singleton phrase,
$\varphi_j = \s[i_{j-1}+1,i_{j}],$ the {\em source} substring's starting position $\ell$ is also chosen together with  the parsing\footnote{Note that in general for a given phrase there could be more than one possible preceding substring which coincide to its longest proper prefix. The assumption that sources are chosen with the parsing $\varphi$ implies that the encoding is uniquely defined by $\varphi$ and allows us to focus only on the parsing.}. 

Given an LZ-parsing $\varphi$ of a string $\s$ an LZ76 encoding of $\s$ is obtained by replacing each phrase of $\varphi_j$ by 
\begin{itemize}
\item the triple $(0,0,s_{i_{j-1}+1})$ if $\varphi_j$ is a singleton phrase, i.e., $i_j = i_{j-1}+1.$
\item the triple $(\ell, i_j-i_{j-1}-1, s_{i_j})$ if $|\varphi_j| > 1$ and $\varphi_j = s[i_{j-1}+1, i_j] = s[\ell, \ell+(i_j-i_{j-1})-2].$
\end{itemize}
\end{definition}

\begin{definition}[Access time - hop-number]
Let $\varphi$ be an LZ-parsing of a string $\s.$ 
We define the access time function 
$hop^{\varphi} : [|\s|] \mapsto \mathbb{N}$
mapping each position $p$ to the number of accesses necessary 
to identify $s_p$ given the LZ-parsing $\varphi$:
$$hop^{\varphi}(p) = \begin{cases}
    0 & \mbox{if $p$ is the position of the last character of a phrase of  $\varphi$}\\
hop^{\varphi}(q)+1 & \mbox{if position $q$ is the source of position $p$ in $\varphi$} 
\end{cases}
$$

We will also refer to $hop^{\varphi}(x)$ as the hop-number of position $x.$

Given an LZ parsing $\varphi$ 
of a string $\s,$ 
we define 
$$hop^{\varphi}_{\max}(\s) = \max_{p = 1, \dots, |\s|} hop^{\varphi}(p).$$
\end{definition}

\begin{definition}[c-BLZ parsing]
Fix an LZ-parsing  $\varphi$ 
of a string $\s.$ 
We say that $\varphi$ is a $c$-BLZ parsing of $\s$ if $hop^{\varphi}_{\max}(\s) \leq c,$ 
i.e., any character of $\s$ can be recovered from the parsing $\varphi$ with at most $c$ hops.
\end{definition}

\noindent{\bf Definitions for graphs.} All graphs considered in this paper are undirected.
A \emph{graph} is a pair $G = (V,E)$ where $V=V(G)$ is a finite set of \emph{vertices} and $E=E(G)$ is a set of $2$-element subsets of $V$ called \emph{edges}. 
Two vertices $u$ and $v$ in a graph $G = (V,E)$ are \emph{adjacent} if $\{u,v\}\in E$. We also say that edge $\{u,v\}$ is incident in vertices $u$ and $v.$
We say that $G$ is $B$-regular if for each vertex $v$ there are exactly $B$ edges incident in $v.$

A \emph{vertex cover} in $G$ is a set $C$ of vertices such that every edge of $G$ is incident in at least one vertex of $C$. The minimum size of
a vertex cover of $G$ is denoted by $\tau(G).$
The {\sc Minimum Vertex Cover (Min-VC)} problem is the following minimization problem: Given a graph $G$ find a vertex cover of $G$ of size $\tau(G).$  We will use {\sc Min-VC-4} for the 
restriction of {\sc MIN-VC} to $4$-regular graphs.

\section{The {\sc BLZ} problem is NP-hard}

For the hardness of the {\sc BLZ Problem} in the introduction, it would be enough to prove the existence of a single $c$ for which the problem is NP-hard. In fact, we are able to show a stronger result, namely, that the hardness of the problem holds for any possible choice of $c$ (constant w.r.t.\  $\s$). For this, we consider the following  variant where the parameter $c$ is fixed (not part of the input).  

\begin{myproblem}
  \problemtitle{$c$-Bounded Access Lempel Ziv Parsing ($c$-BLZ)}
  \probleminput{A strings $\s$.}
  \problemquestion{An LZ-parsing $\varphi^*_c$ of $\s$ with the minimum number of phrases among all LZ-parsings $\varphi_c$  
  of $\s$ satisfying $hop_{\max}^{\varphi_c}(\s) \leq c$}
\end{myproblem}

We show that for any integer $c\geq 1$ the {\sc $c$-BLZ} problem is NP-hard. 
This implies the  hardness of the {\sc BLZ Problem} in the introduction.

\subsection{The reduction} \label{sec:reduction}
Let $G = (V, E)$ be a graph. We are going to show how to build a string $\s = \s_G$ over an alphabet of size $O(m+n),$ where $n = |V|$ and 
$m = |E|.$ The string $\s_G$ will be constructed in order to 
guarantee that $G$ has a vertex cover of size $k$ iff the string 
$\s_G$ has a $c$-BLZ parsing of size $4n + 6m + k + \ell(c-1),$ where $\ell$ is a constant $> k.$ Since finding a minimum vertex cover is NP-hard, from our reduction the  NP-hardness of the {\sc $c$-BLZ} problem follows.

Let $v_1, \dots, v_n$ denote the vertices of $V$ and $e_1, \dots, e_m$ denote the edges in $E.$
The string $\s$ will be the concatenation of several pieces that we are now going to describe.

Let $P$ be the following string that contains a distinct character for each vertex and a distinct character for each edge of $G$ separated by additional and distinct {\em separator}-characters $\#^p_1, \dots, \#^p_{n+m}.$ 
\[P = v_1\, \#^p_1 \, v_2\, \#^p_2 \dots v_n\, \#^p_n \, e_1 \, \#^p_{n+1} \, e_2 \, \#^p_{n+2} \, \dots e_m \, \#^p_{n+m} \]

For each $i = 1, \dots n,$ we have a ({\em vertex}) string $X_i$ associated to vertex $v_i,$ which is defined as follows:
   
\[X_i = v_i' \, v_i\, \#^v_i.\]

Moreover, we define $X$ to be the concatenation of these vertex strings:
\[X = X_1 X_2 \dots X_n.\]

For each $i = 1, \dots m,$ we have an ({\em edge}) string $Y_i$ associated to edge $e_i = (v_p, v_q),$ which is defined as follows: 
\begin{equation} \label{eq:Yi}  
Y_i = v_p' \, v_p \, e_i \, \$_i \, v_q' \, v_q \, e_i \, \$_i \,  \#^e_i .
\end{equation}
Moreover, we define $Y$ to be the concatenation of these edge strings:
\[Y = Y_1 Y_2 \dots Y_m.\]

Let $\s^{(1)} = P \, X \, Y$

For any integer $\ell > k,$ we recursively define the strings $\alpha^{(i)},$ for $i=1, 2, \dots$

$$\alpha^{(i)} = 
\begin{cases}
\s^{(1)} = PXY & i = 1\\
\alpha^{(i-1)} \, \beta^{(i-1)} & i > 1
\end{cases}
$$
where for $i \geq 1$
\[\beta^{(i)} = 
\alpha^{(i)} \, \#^{(i)}_1 \, \dots \, \alpha^{(i)} \, \#^{(i)}_{\ell},\]

Then, we let the string $\s = \s_G$ be defined by: 

\begin{multline*}
\s = \alpha^{(c)} =  \alpha^{(1)} \, \beta^{(1)} \, \beta^{(2)} \, \dots \, \beta^{(c-1)} \hfill{~} \\ = 
\underbrace{\underbrace{\underbrace{\alpha^{(1)} \beta^{(1)}}_{\alpha^{(2)}} \,
\underbrace{ \alpha^{(2)} \, \#^{(2)}_1 \dots \alpha^{(2)} \, \#^{(2)}_{\ell}}_{\beta^{(2)}} }_{\alpha^{(3)}}
\underbrace{ \alpha^{(3)} \#^{(3)}_1 \dots \alpha^{(3)} \#^{(3)}_{\ell}}_{\beta^{(3)}}     }_{\alpha^{(4)}} \dots 
\underbrace{ \alpha^{(c-1)} \#^{(c-1)}_1 \dots \alpha^{(c-1)} \#^{(c-1)}_{\ell}}_{\beta^{(c-1)}} 
\end{multline*} 

For each $i = 1, \dots, c-1,$ and for  $\gamma = 1, \dots, \ell,$ 
we denote by $\alpha^{(i)}_{\gamma}$ the $\gamma$th copy of $\alpha^{(i)}$ that appears as a substring in
$\beta^{(i)},$ i.e., the substring of $\beta^{(i)}$ of size $|\alpha^{(i)}|$ immediately preceding the unique occurrence in $\beta^{(i)}$ of the character 
$\#^{(i)}_{\gamma}.$

\medskip

\noindent
{\bf The intuition behind the reduction.}
The basic idea of the reduction is based on the fact that for any parsing of the string $\s^{(1)} = PX_1 \dots X_n Y_1 \dots Y_m$ the following holds:
\begin{itemize}
\item there is a unique way to parse $P,$ 
\item for each $X_i$ there are two ways to parse it, into $3$ or $2$ phrases; which will, respectively, encode the inclusion or not of the vertex $v_i$ into a vertex cover for $G$
\item for a $Y_j$ corresponding to edge $e_j = (v_i, v_{i'}),$ $Y_j$ can be parsed into $4$ phrases if and only if at least one between $X_i$ and $X_{i'}$ has been parsed into $3$ phrases, which encodes the fact that $e_j$ is covered only if one of its incident vertices are in the vertex cover. Otherwise $Y_j$ must be parsed into $\geq 5$ phrases.
\end{itemize} 
The bound on the parsing set by the reduction can then be attained only if each $Y_j$ is parsed into $4$ phrases, which by the above properties, implies that the set of $X_i$'s parsed into $3$ phrases must be a vertex cover of $G.$

\medskip
The next theorem formalizes the relationship between minimum vertex covers of $G$ and 
shortest $1$-BLZ parsing of the prefix $\s^{(1)}$ of $\s_G.$ We note that this theorem 
is already enough to establish that $1$-{\sc BLZ} is NP-hard, which implies also the hardness of the {\sc BLZ} problem in the introduction.

\begin{theorem} \label{thm:case1}
Fix a graph $G = (V,E)$ and a non-negative integer $k,$ and let $\s$ be the string produced by the construction above. Then, 
$G$ has a vertex cover of size $k$ if and only if there exists a $1$-BLZ parsing $\varphi$ for $\s^{(1)}$ of size $4n+6m+k.$
\end{theorem}
\begin{proof}
String $P$ can be parsed in a unique way with each character being a distinct phrase, since each character is different from any other one, hence we have
$$\varphi(P) = \mid v_1 \mid \#^p_1 \mid v_2 \mid \#^p_2 \mid \dots \mid v_n \mid \#^p_n \mid e_1 \mid \#^p_1 \mid e_2 \mid \#^p_1 \mid \dots 
\mid e_m \mid \#^p_m \mid.$$

For each $X_i = v'_i \, v_i \,  \#^v_i$ there are only two ways to parse it, satisfying the bound that the max-hop-number is $\leq 1,$ namely
\begin{equation} \label{vi-selected}
v'_i \mid v_i \mid \#^v_i \mid
\end{equation}
\begin{equation} \label{vi-not}
v'_i \mid v_i \, \#^v_i \mid
\end{equation}

To see this, observe the occurrence of $v'_i$ in $X_i$ is the first occurrence of this character in $\s^{(1)},$ while $v_i$ already occurred in $P.$ 

We refer to the case (\ref{vi-selected}) by saying that the parsing {\em selects} the vertex $v_i.$ 
Moreover, we also say that $\varphi$ selects vertex $v_i$ if the pair of adjacent phrases $\mid v_i' \mid v_i \mid$ appears somewhere in $\varphi(\s^{(1)}),$ with result that 
the hop-numbers of the positions of these occurrences of $v_i$ and $v_i'$ are both $0.$

Conversely, if the pair of consecutive phrases $\mid v_i' \mid v_i \mid$ does not appear in $\varphi$ we 
say that the parsing {\em does not select} the vertex $v_i.$

\medskip

\noindent
{\em Claim 1.}
Without loss of generality, we can assume that if a parsing selects $v_p$ then, it parses the substring $X_p$ (in $\s^{(1)}$) into three phrases as in (\ref{vi-selected})
$$\varphi(X_p) = \mid v_p'\mid v_p \mid \#^v_p \mid$$

\begin{proof}
It is enough to observe that if a parsing $\varphi$ selects $v_p$ and does not parse $X_p$ into three phrases, then there is an edge substring $Y_i,$ corresponding to some edge $e_i$ incident to $v_p,$ where $\varphi$ produces  the two consecutive phrases $\mid v_p' \mid v_p \mid.$ 
Consider the parsing $\varphi'$ that: (i) 
parses $X_p$ into three phrases and uses a single 
phrase for the substring $Y'_i = v_p' v_p$ in $Y_i$; and 
for any phrase $\phi$ that uses the substring $Y'_i = v'_p v_p,$ of $Y_i,$ as a source (where both positions have $hop = 0$ in $\varphi$) we make the substring $v'_p v_p$ of $X_i$ to be the source of $\phi$ in $\varphi'$. 
It follows that $|\varphi'| = |\varphi|$ and
$$ hop^{\varphi'}_{\max}(\s) = hop^{\varphi}_{\max}(\s)$$
and $\varphi'$ satisfies the claim with respect to the way $v_p$ is selected. 
\qed
\end{proof}
\medskip

Let us now consider a substring $Y_i$ corresponding to the edge $e_i = \{v_p, v_q\}.$
Recall the definition of $Y_i$ from (\ref{eq:Yi}).
\medskip

\noindent
{\em Claim 2.} Every $1$-BLZ parsing of $\s^{(1)}$ parses $Y_i$ into $\geq 4$ phrases.
\begin{proof}
Since before its first occurrence in $Y_i$ the character $e_i$ only appeared in $P$ and  it was not preceded by $v_p$ no phrase 
containing $v_p$ can contain also $\$_i.$ Moreover since the first occurrence of $\$_i$ is the first occurrence of such character in $\s^{(1)}$
every parsing will have a phrase ending here. Therefore, the substring $v'_p v_p e_i \$_i$ must be parsed in at least two phrases.
Analogously, no phrase can contain both $v_q$ and $\$_i,$ hence also the substring $v_q' v_q e_i \$_i \#_i^e$ requires at least 2 phrases.    \qed
\end{proof}

\medskip

\noindent
{\em Claim 3.} A $1$-BLZ parsing $\varphi$ of $\s^{(1)}$  can parse $Y_i$ into $4$ phrases if and only if 
$\varphi$ selects at least one of the vertices $v_p, v_q.$
\begin{proof}
If both vertices $v_p$ and $v_q$ are selected a valid $1$-BLZ parsing of $Y_i$ into exactly $4$ phrases is:
$$\varphi(Y_i) = 
\underset{1~~1~~0~~~0~~~~1~~~1~~0~~~1~~0~}{v_p' \, v_p \, e_i \mid \$_i \mid v'_q \, v_q \, e_i \mid \$_i \, \#^e_i},$$
where the number underneath each character is the 
hop-number for the corresponding position. These follow since $\varphi(X_p) = \mid v'_p \mid v_p \mid \#^v_p \mid, \, \varphi(X_p) = \mid v'_p \mid v_p \mid \#^v_p \mid$  provide sources $v'_p v_p$ and $v'_q v_q$ where all positions have hop-number 0.

If only vertex $v_p$ is selected, from $\varphi(X_p)$ and $\varphi(X_q)$ we have a source $v'_p v_p$ where both positions have hop-number $0$ and the source $v'_q$ with hop-number 0. Hence, a valid $1$-BLZ parsing for $Y_i$ is:
$$\varphi(Y_i) = 
\underset{1~~1~~0~~~~0~~~~1~~~0~~~~0~~1~~0~}{v_p' \, v_p \, e_i \mid \$_i \mid v_q' \, v_q \mid e_i \, \$_i \, \#^e_i}, 
$$
where the source of $\mid e_i \, \$_i \, \#^e_i \mid $ is the previous occurrence of $e_i \, \$_i$  in $Y_i.$
Analogously, if only the vertex $v_q$ is selected a valid $1$-BLZ parsing for $Y_i$ is:
$$ 
\varphi(Y_i) = \underset{1~~0~~~1~~0~~~~1~~~1~~0~~~1~~0~}{v_p' \, v_p \mid e_i \, \$_i \mid v_q' \, v_q \, e_i \mid\$_i \,  \#^e_i},
$$
with the sources: for $\mid e_i \, \$_i \mid$ from $P$ and for  $\mid\$_i \,  \#^e_i \mid$  from previous $\$_i$ in $Y_i.$

It remains to be shown that if the $1$-BLZ parsing $\varphi$ does not select any of the vertices $v_p, v_q$, then it must parse  $Y_i$ into $\geq 5$ phrases. 
Since the vertex $p$ is not selected there is no substring $v_p' \, v_p$ where both corresponding positions have hop-number $0$. Because of this in the first half of $Y_i,$ that encodes the incidence in the the vertex $p$, namely, the substring $v_p' \, v_p \,e_i \, \$_i $, no phrase containing $v_p' \, v_p$ can contain also $e_i,$ for otherwise the position of $v_p'$ would have hop-number 2. Analogously since also the vertex $q$ is not selected no phrase containing $v_q' \, v_q$ can contain also $e_i$. Additionally, as reported in the previous claim, the first occurrence of $\$_i$ and $\#^e_i$ need to be positioned at the end of a phrase. These observations imply that for any valid $1$-BLZ parsing of $Y_i$ there are only two options left: (i) the parsing splits the two adjacent characters $e_i\, \$_i$ into two singleton phrases giving each one of them hop-number $0$ and uses a single phrase for the suffix $\mid e_i \$_i \#^e_i \mid$. The resulting parsing is then 
$$ 
\underset{1~~0~~~~0~~~0~~~~1~~~0~~~1~~1~~0~}{v_p' \, v_p \mid e_i \mid \$_i \mid v_q' \, v_q \mid e_i  \$_i \,  \#^e_i}
$$
with $5$ phrases. 
(ii) Alternatively, the substring $e_i\, \$_i$ in the first half is parsed as a single phrase. In this case the two positions have hop-number $1$ and $0,$ respectively, as the only occurrence of $e_i$ is in $P$ and it is not followed by $\$i$. In this case the second occurrence of $e_i$ can't be parsed in the same phrase of $\#^e_i$ and two different phrases are needed for the last three character of $Y_i$, namely we have the following two possible parsings:
$$\underset{1~~0~~~1~~0~~~~1~~~0~~~0~~~1~~0~}{v_p' \, v_p \mid e_i \, \$_i \mid v_q' \, v_q \mid e_i \mid \$_i \,  \#^e_i} \qquad 
\underset{1~~0~~~1~~0~~~~1~~~0~~~1~~0~~~0~}{v_p' \, v_p \mid e_i \, \$_i \mid v_q' \, v_q \mid e_i \, \$_i \mid  \#^e_i}
$$
both leading to a total number of phrases equal to $5.$ 

\qed
\end{proof}
We can now complete the proof of the theorem by considering the two directions of the equivalence in the statement.

\noindent
{\em 1. If $\tau(G) \leq k$ then there exists a $1$-BLZ parsing of $\s$ of size $\leq 4n+6m+k.$} Let us now assume that there is a vertex cover $C$ of $G$ of size $\leq k.$ Consider the parsing $\varphi_C$ that in $X$ {\em selects} each vertex in $C.$ Such parsing, by Claim 2, will have: $4$ phrases for each edge substring $Y_i$; $2$ phrases for each $X_i$ corresponding to a non-selected vertex $v_i$; and (by Claim 1) $3$ phrases for each $X_i$ corresponding to a selected vertex $v_i.$ Considering also the $2n+2m = |P|$ phrases necessary to parse $P,$ in total, the parsing $\varphi_C$ will have size
$$|\varphi_C| = |P| + 2|V| + |C| + 4 |E| \leq 4n+6m + k,$$
where for the last inequality we use $|C| \leq k.$

\bigskip
\noindent
{\em 2. If $\tau(G) > k$ then any $1$-BLZ parsing of $\s$ has size $> 4n+6m+k.$}
Recall that we say that a parsing $\varphi$ selects a vertex $v_p$ if it produces the sequence of two adjacent phrases $\mid v_p'\mid v_p\mid$, hence with both positions having hop-number $0,$ either in the substring $X_p$ or in some substring $Y_i$ corresponding to an edge incident to $v_p$. 
By Claim 1 we can assume that selection of a vertex $v_p$ is only done in the string $X_p.$

Let $k' = \tau(G) > k$. Given $\varphi$ there are two possible scenarios: 

\noindent
{\em Case 1.} $\varphi$ selects more than $k$ vertices. 
Let $h$ be the number of vertices selected by $\varphi$. By Claim 1, for each selected vertex  $v_i$ the corresponding $X_i$ is parsed in $3$ phrases so the number of phrases used by $\varphi$ over all $X_i$ is $|\varphi(X)| = 2n+h.$ By Claim 2, the number of phrases used to parse all the $Y_i$ is at least $4m$. Then,    
$|\varphi| = |\varphi(P)|+|\varphi(X)|+|\varphi(Y)| \geq 4n+6m + h > 4n+6m+k,$
where for the last inequality we use the hypothesis $h>k$.

\medskip

\noindent
{Case 2.} $\varphi$ selects $\leq k$ vertices. Let $C$ be the set of such selected vertices. Then there exist $t \geq k' - |C|$ edges not covered by $C.$ We can affirm this because, it's possible to build a vertex cover from $C$ adding a vertex for each edge that is not covered in $C$, this vertex cover would have a size equal to $|C| + t$ that need to be $\geq k'$ since $k'$ is the size of a minimum vertex cover in $G$.
Each one of the $t$ edges not covered corresponds to a $Y_i$ that is parsed by $\varphi$ in at least $5$ phrases, by  Claim 3. The total number of phrases used to parse all these $Y_i$s is then $\geq 5\,t$, hence the number of phrases used to parse $Y = Y_1 \cdots Y_m$ is at least $4m + t.$ Moreover, the number of phrases used to parse all of $X = X_1 \cdots X_n$ is at least $2n + |C|.$ Putting together these two bound and the fact that $|\varphi(P)| = 2n+2m,$ we conclude that
$$|\varphi| = |\varphi(P)| + |\varphi(X)|+ |\varphi(Y)| \geq 4n+6m + t + |C| \geq 4n+6m+k' > 4n+6m+k,$$
where we are using $t+|C| \geq k'$ and $k' > k $.
\qed
\end{proof}

\medskip

The following  theorem extends the result of Theorem \ref{thm:case1} to the case of $c > 1.$ The idea is to show that any $c$-BLZ $\varphi$ parsing for $\s_G$ that has size $\leq 4n+6m+k+(c-1)\ell$ when restricted to $\s^{(1)}$ has to be a $1$-BLZ parsing of $\s^{(1)}$ of size $4n+6m+k,$ hence encoding a vertex cover of $G$ of size $k.$  The proof is deferred to the appendix.

\begin{theorem} \label{thm:minVC-minParsing}
Fix a graph $G=(V,E)$ and let $n = |V|, m = |E|.$ Fix integers $c \geq 1$ and $0 \leq k \leq n,$ and $\ell > k.$
Let $\s_G$ be the string produced by the procedure in section \ref{sec:reduction}. Then 
$G$ has a vertex cover of size $k$ if and only if there is a $c$-BLZ parsing of $\s_G$ of size $4n+ 6m + k + \ell(c-1).$
\end{theorem}

We are ready to state the state the main result of this section.
\begin{theorem}
For any integer $c > 1$ the {$c$-BLZ} problem is NP-hard. 
\end{theorem}
\begin{proof}
The problem of deciding for a given graph $G$ and a given $k$ whether $G$ has a vertex cover of size $k$ is well known to be NP-hard (see, e.g., \cite{CHLEBIK2006320}). By Theorem \ref{thm:minVC-minParsing} the polynomial time function that
maps a graph $G$ and integers $k$ and $\ell > k$ to the string $\s_G$ is a polynomial time reduction between the problem of deciding whether $G$ has a vertex cover of size $k$ and the problem  of deciding whether $\s_G$ has a $c$-BLZ parsing of size $4n+6m+(c-1)\ell.$ It follows that the latter problem is also NP-hard.
\end{proof}

\section{APX Hardness}
In this section we focus on the existence of polynomial time approximation algorithms for the problem of finding an optimal $c$-BLZ parsing of a string. We show that the problem is $APX$-hard, hence  
there exists a constant $\kappa$ such that no polynomial time
$\kappa$-approximation algorithm can exists for the {\sc BLZ} problem 
unless $P = NP.$

For the proof of APX hardness we will employ the concept of an L-reduction: 
Let $\mathbb{A}, \mathbb{B}$ be two minimization problems. 
$\mathbb{A}$ is said to be  L-reducible to $\mathbb{B},$ denoted by 
$\mathbb{A} \leq_L \mathbb{B},$ if two polynomial time computable functions $f$, $g$ and two constants $a$, $b$ exist such that for any 
instance $\chi$ of $\mathbb{A}$:
\begin{enumerate}
\item $f(\chi)$ is an instance of $\mathbb{B}$
\item $OPT_{\mathbb{B}}(f(\chi)) \leq a \, OPT_{\mathbb{A}}(\chi)$
\item for any solution $s_{\mathbb{B}}(f(\chi))$ to the instance $f(\chi)$ of $\mathbb{B},$ $g(s_{\mathbb{B}}(f(\chi)))$ is a 
solution to the instance $\chi$ of $\mathbb{A}.$
\item for any solution $s_{\mathbb{B}}(f(\chi))$ to the instance $f(\chi)$ of $\mathbb{B},$ it holds that
$$ cost(g(s_{\mathbb{B}}(f(\chi)))) - OPT_{\mathbb{A}}(\chi) \leq 
b \left( cost(s_{\mathbb{B}}(f(\chi)))) - OPT_{\mathbb{B}}(f(\chi)) \right), $$
\end{enumerate}
where $OPT_{\mathbb{P}}(x)$ denotes the value of the optimal solution to instance $x$ of problem $\mathbb{P}$ and 
$cost_{\mathbb{P}}(s)$ denotes the cost in problem $\mathbb{P}$ of a solution $s.$
When conditions $1$-$4$ are satisfied, the quadruple $(f,g,a,b)$ is  an L-reduction from  $\mathbb{A}$ a to $ \mathbb{B}$.

\begin{theorem}
The problem $c$-{\sc BLZ} is APX-hard.
\end{theorem}
\begin{proof}
We prove that there exists an L-reduction $(f,g,a,b)$ from {\sc Min-VC-$4$}, the problem of finding a minimum vertex cover in a $4$-regular graph $G=(V,E)$ to 
the minimization problem c-BLZ on the string $s_G$.
Since {\sc Min-VC-$4$} is APX-hard \cite{CHLEBIK2006320}, 
the existence of the L-reduction implies that c-BLZ is also APX-hard \cite{612321} (also see \cite{Alimonti-Kann97,PAPADIMITRIOU1991425}.\\

The function $f$, that associate instances of the vertex cover problem to instances of c-BLZ problem, in this L-reduction is going to be the same one used for the proof of NP-hardness.\\
Since we are working with $4$-regular graph we know that $m=|E|=2|V|=2n $ and  $\frac{n}{2} \leq \tau (G),$ 
where $\tau (G)$ is the size of a  minimum vertex cover of $G$.

Let $k^* = \tau(G) = OPT_{\sc VC}(G).$ and $\varphi^*$ denote a $c$-BLZ parsing of $\s_G$ of minimum size, i.e., $|\varphi^*| = OPT_{BLZ}(\s_G, c).$ 
By Theorem \ref{thm:minVC-minParsing}, with $\ell = n > k^*,$ we have 
$$|\varphi^*| = 4n + 6m +k^*+\ell(c-1) = 16 n+k^*+\ell(c-1). $$
This holds for each $\ell \geq k$ where $k $ is the vertex cover's size. In particular, it holds for $\ell = n$. 
Using $\ell = n$ and the property of $|\tau(G)| \geq \frac{n}{2}$ from 
the previous inequality  we get
$\displaystyle{|\varphi^*| = 15 n + nc + k^* \leq 30 k^* + k^* +2 k^*c = (31+ 2c)k^*.}$
Since $c$ is a constant, setting $a = (31+2c),$ we have that 
our reduction satisfies the first two items in the definition of an $L$-reduction.

We define the function $g$ that maps a parsing  $\varphi$ of the string $s_G$ to a vertex cover $C_\varphi$ of $G$ by providing the following algorithm to compute it:  start with $C_\varphi$ being the empty set; \\
Given $\varphi$, parsing of the the string $s_G$, 
for each $X_i$ of $\alpha^{(1)},$ if the parsing $\varphi$ selects $X_i,$ i.e., it parses it into 3 phrases, then $v_i$ is added to $C_\varphi$. Let $d$ be the number of $X_i$ selected in $\varphi,$ hence at this point we have $|C_{\varphi}| = d$.\\
If $C_{\varphi}$ is a vertex cover for $G$ then then the computation of $g$ terminates. Otherwise, we look for every $Y_i$ in $\alpha^{(1)}$ which is either  parsed into $\geq 5$ phrases or it includes a position $x$ with $hop(x)=2$ and $4$ phrases. We refer to the  number of $Y_i$s in $\varphi$ that have such characteristic as $p$. For each one of these $Y_i,$ let $e_i$ be the corresponding edge of $G$ and add to $C_\varphi$ one of the vertices incident to $e_i.$ The resulting $C_\varphi$ is a vertex cover for $G$, as we know that each $Y_i$ not considered in the last step is parsed into 4 phrases and does not have a position with hop-number $> 1.$ 
Let the edge of $G$ corresponding to $Y_i$ be $e_i = \{v_p, v_q\}.$
Then by Claim 2 in Theorem \ref{thm:case1}, one of $v_p$ and $v_q$ is selected by $\varphi,$ hence in $C_{\varphi}.$ 

Then, we set $g(\varphi) = C_\varphi.$ 

Since $C_\varphi$ is a vertex cover for $G$ then there exists a parsing $\varphi '$ that: selects each vertex $v_i \in C_\varphi$, parses each $Y_i$ in $\alpha^{(1)}$ using exactly $4$ phrases and parses the remaining string: $\beta^{(1)}, \dots \beta^{(c-1)}$  with $\ell(c-1)$ additional phrases.

The number of phrases used in this parsing satisfies 
$$ |\varphi'| = |C_\varphi| + 4n + 6m + \ell(c-1). $$

Since $|C_\varphi| \leq d+p$ we have that   
$|\varphi| \geq 4n + 6m + d + p + \ell(c-1) \geq  \varphi'.$
From Theorem \ref{thm:minVC-minParsing}, we also have 
$|\tau(G)| = |\varphi^*| - (4n + 6m+ \ell(c-1)).$
Then, we have
$\displaystyle{|C_\varphi| -|\tau(G)|\leq |\varphi| -|\varphi^*|, }$
which implies that the function $g$ defined above  and $b=1$ satisfy property 3 and 4 of an $L$-reduction.

We conclude that the quadruple $(f, g, a, b)$
is the desired $L$-reduction from {\sc Min-VC-$4$} to {\sc $c$-BLZ}. It follows that the latter problem is APX-hard. 
\qed
\end{proof}

\section{On the approximation of BLZ to LZ76}
For a string $\s,$ we denote by $OPT_{LZ}(s)$ the size of an optimal LZ parsing for $\s$ and by $OPT_{c-BLZ}(\s)$ 
the size of an optimal $c$-BLZ parsing for $\s.$

Trivially we have $OPT_{LZ}(s) \leq OPT_{c-BLZ}(\s).$ Before our results, $OPT_{LZ}(s)$ was the only known and used lower bound on the size of
$OPT_{c-BLZ}(\s)$\cite{LMN24,bannai2024height}.

In this section we study the ratio between the size of an  optimal LZ parsing and the size of an optimal $c$-BLZ parsing of the same string $\s.$ We will show that in the worst case the ratio is unbounded, by exhibiting a class of strings over a ternary alphabet for which the 
optimal LZ parsing is known to be logarithmic in the size of the string while the optimal $c$-BLZ parsing is $\Omega(n^{\frac{1}{c+1})}.$ 

Remarkably, the latter result provide the first non trivial lower bound on the optimal $c$-BLZ parsing of a string which might be significantly larger than the size of the LZ-parsing. 
More precisely, for the size of an optimal $c$-BLZ parsing of a square free string we prove the following lower bound. 

\begin{lemma} \label{lemma:bound} 
Let's fix an integer $c \geq 1$ and a \textit{square free} string $\s$, then for every $c$-BLZ parsing $\varphi$ of $s$ it holds that
$\displaystyle{| \varphi(\s)| \geq \sqrt[c+1]{|\s|} -1.}$
\end{lemma}
\begin{proof}
We start by observing that if the string is square free then the source of a phrase cannot overlap with it.  

We argue by contradiction. We assume $|\varphi(\s)| < \sqrt[c+1]{n} -1.$ Then there exists a phrase $\phi^{(1)},$ such that: 
\begin{equation} \label{fi-1}
 |\phi^{(1)}| > \frac{n}{\sqrt[c+1]{n} -1} > n^{\frac{c}{c+1}} +1.
\end{equation}
For $i = 1, \dots, c,$ we define strings $\phi^{(i)}$ and $\Tilde{\phi}^{(i)}$ as follows: (i) $\Tilde{\phi}^{(i)}$ is the source of $\phi^{(i)}$; 
and $\phi^{(i+1)}$ is the largest substring of $\Tilde{\phi}^{(i)}$ which is contained in a phrase of $\varphi.$
Then, inductively we can show that for each $i=1, \dots, c$ we have 
$$ |\phi^{(1)}| > n^{\frac{c}{c+1}} +1, \qquad \mbox{hence} \qquad |\Tilde{\phi}^{(1)}|>n^{\frac{c}{c+1}}.$$
\noindent
{\em Base case $i=1.$} From (\ref{fi-1}) we immediately also have $\displaystyle{|\Tilde{\phi}^{(1)}|>n^{\frac{c}{c+1}}.}$

\smallskip

\noindent
{\em Inductive step $i>1.$} Since, by the induction hypothesis on $i-1$, we have $\displaystyle{|\phi^{(i)}| \geq \frac{\Tilde{\phi}^{(i-1)}}{\sqrt[c+1]{n} -1} > \frac{n^{\frac{c-(i-1)+1}{c+1}} +1}{n^ \frac{1}{c}-1}> n^{\frac{c-i+1}{c+1}} +1,}$
it holds that:
$$|\phi^{(i)}| > n^{\frac{c-i+1}{c+1}} +1 \quad \mbox{hence } |\Tilde{\phi}^{(i)}|>n^{\frac{c-i+1}{c+1}}.$$
For each $i,$ the maximum hop-number of the positions of $\phi^{(i)}$ is $\leq c-i+1$ therefore for $i=c,$ we have that the maximum hop-number of the positions of $\phi^{(c)}$ is $\leq 1.$ Then, the 
sources of these positions (i.e., the positions of $\Tilde{\phi}^{(c)}$) must be consecutive positions whose hop-number is 0. This can only be the case if they are the positions of single character phrases. Hence: $$|\varphi| \geq |\Tilde{\phi}^{(c)}| > n^{\frac{c-c+1}{c+1}} = n^{\frac{1}{c+1}}.$$
This contradicts the initial (absurdum) hypothesis and completes the proof. 
\qed
\end{proof}

From the results in \cite{constantinescu2007lempel,leech19572726} we have the following upper bound on $OPT_{LZ}(\s).$ 
\begin{fact} \label{fact:leech}
There are infinitely many string $s {\in \{1,2,3\}}^*$, which are \textit{square free} and such that:
$OPT_{LZ}(s) = O(\log{|s|}).$
\end{fact}

\begin{theorem}
Fix an integer $c\geq 1.$ Then, the worst case ratio between 
$OPT_{c-BLZ}(\s)$ and $OPT_{LZ}(\s)$ over all ternary strings $\s$ is unbounded, i.e., 
$$  
\max_{\s \in {\{1,2,3\}}^*}  \frac{OPT_{c-BLZ}(\s)}{OPT_{LZ}(\s)}
\to \infty.$$

\end{theorem}
\begin{proof}
Let $\mathcal{S}$ be the set of square free ternary string. Then, by Fact \ref{fact:leech} and Lemma \ref{lemma:bound}, there exists a constant $a$ such that:
$$  
\max_{\s \in \mathcal{S}}  \frac{OPT_{c-BLZ}(\s)}{OPT_{LZ}(\s)}\geq \frac{|\s|^{\frac{1}{c+1}}-1}{a \log|\s|} {\to} \infty.$$
\qed
\end{proof}
\noindent
{\em Remark: Since there exist $\s$ s.t. $OPT_{LZ}(\s) \geq |\s|^{\frac{1}{c+1}}$ by Lemma \ref{lemma:bound} we obtain the following general lower bound %
$OPT_{c-BLZ}(\s) \geq \max\{OPT_{LZ}(\s), |\s|^{\frac{1}{c+1}}\}.$}

\section{Conclusion and Open Problems}
We studied 
{\sc $c$-BLZ}, a variant of the LZ76,  that allows to decompress characters of the text in a bounded number of accesses to the encoding.
, i.e., without needing to decompress the whole text.
We proved that for any constant $c$ computing the optimal parsing that guarantees decompression of a character with at most $c$  accesses is NP-hard and also APX hard. We also showed that the ratio to the size of the optimal LZ76 parsing is unbounded in the worst case, providing a first non-trivial lower bound on the size of the optimal $c$-BLZ parsing. 

A main direction for future research is the investigation of approximation algorithms as well as parameterized algorithms for bounded parameters of practical importance, e.g., alphabet size. The unbounded ratio to the optimal LZ76 parsing leaves open the question of lower bounds on $OPT_{c-BLZ}.$

\newpage

 \bibliographystyle{splncs04}
 \bibliography{BLZc}

\newpage

\appendix

\section{The proof of Theorem \ref{thm:minVC-minParsing}}

{\bf Theorem \ref{thm:minVC-minParsing}.}
{\em Fix a graph $G=(V,E)$ and let $n = |V|, m = |E|.$ Fix integers $c \geq 1$ and $0 \leq k \leq n,$ and $\ell > k.$
Let $\s_G$ be the string produced by the procedure in section \ref{sec:reduction}. Then 
$G$ has a vertex cover of size $k$ if and only if there is a $c$-BLZ parsing of $\s_G$ of size $4n+ 6m + k + \ell(c-1).$
}

The proof is based on several lemmas.

\begin{lemma} \label{lemma:VC2Parsing}
If $G$ has a vertex cover of size $k$ then there exists a $c$-BLZ parsing of $\s$ of size $4n+6m+k+ (c-1)\ell.$ 
\end{lemma}
\begin{proof}
By Theorem \ref{thm:case1}, there is $1$-BLZ parsing $\varphi^{(1)}$ of $\alpha^{(1)}$ of size $4n+6m+k.$ 
Let $\varphi$ be the parsing of $\s$ that uses exactly the same phrases as
$\varphi^{(1)}$ on $\alpha^{(1)}$ and then parses each $\beta^{(i)}$ using 
one phrase for each substring $\alpha^{(i)}_{j} \#^{(i)}_j,$ i.e.
$$\varphi(\beta^{(i)}) = 
\mid \alpha^{(i)}_{1} \#^{(i)}_1\mid \alpha^{(i)}_{2} \#^{(i)}_2 \mid
\cdots \mid \alpha^{(i)}_{\ell} \#^{(i)}_{\ell} \mid,$$
where in each phrase, the prefix $\alpha^{(i)}_{j}$ is referred back to the 
occurrence of $\alpha^{(i)}$ preceding $\beta^{(i)}.$
Hence we have that the overall parsing uses $4n+6m+k+\ell(c-1)$ phrases. 

Moreover, we have that for each $i=1, \dots, c-1$
$$hop^{\varphi}_{\max}(\alpha^{(i+1)}) = hop^{\varphi}_{\max}(\beta^{(i)}) = hop^{\varphi}_{\max}(\alpha^{(i)}) + 1$$ which implies that
$$hop^{\varphi}_{\max}(\s) = hop^{\varphi}_{\max}(\alpha^{(c)}) = hop^{\varphi}_{\max}(\beta^{(c-1)}) = 
hop^\varphi_{\max}(\alpha^{(1)}) + c-1 = c.$$
\end{proof}

\begin{lemma} \label{lemma:LB2-alpha1}
Assume that $\tau(G) = k' > k$ and let $\varphi$ be a $c$-BLZ parsing of $\s$ that on $\alpha^{(1)}$ uses 
$4n+6m+t$ phrases for some $t \leq k.$ Then, there are at least $z = k' - t$ positions in distinct $Y_j$ substrings of $\alpha^{(1)}$ with hop-number  $\geq 2.$   
\end{lemma}
\begin{proof}
Let $\varphi^{(1)}$ be the parsing of $\alpha^{(1)}$ which uses exactly the same phrases as $\varphi.$
Then, $|\varphi^{(1)}| \leq 4n+6m +k.$ 
By Theorem \ref{thm:case1} if there is no vertex cover of size $k$ in $G$ then $\varphi^{(1)}$ cannot be a 
$1$-BLZ parsing of $\alpha^{(1)},$ i.e., 
for some position the  hop number induced by $\varphi$ is
$>1.$

By Claim 3 in the proof of Theorem \ref{thm:case1}, 
we can  assume, w.l.o.g., that if $\varphi^{(1)}$
selects a vertex $v_p$ then it parses 
$X_p$ into three phrases .
Let us denote by $C_{\varphi}$ the set of 
vertices that are selected by $\varphi^{(1)}.$

We have $|\varphi^{(1)}(P)| = 2n + 2m$ and by Claim 1 in the proof of Theorem \ref{thm:case1}, 
$|\varphi^{(1)}(Y)| \geq 4m$ and 
$\varphi^{(1)}(X) = 2n + |C_{\varphi}|.$
Hence, together with 
$|\varphi^{(1)}(\alpha^{(1)})| \leq 4n+6m + k$ it follows that $|C_{\varphi}| \leq k.$
Then, $C_{\varphi}$ is not a vertex cover and in particular, since there are no vertex covers of size $< k',$ there are at least $k' - |C_{\varphi}|$
edges $e_i = (v_p, v_q)$ of $G$ where 
neither $v_p$ nor $v_q$ are in $C_{\varphi}.$
By Claim 3 in the proof of Theorem \ref{thm:case1}, if $\varphi^{(1)}(Y_i) = 4$ there is a position in $Y_i$ with hop number $>1.$
Let $g$ be the number of $Y_i$ parsed into $5$ phrases and $z$ be the number of 
$Y_i$ parsed into $4$ phrases and having a position with $hop \geq 2.$ 

We have $g + |C_{\varphi}| \leq t.$ 
Then, from $g+z \geq k' - |C_{\varphi}|$ we get the desired result  
$$z \geq k' - |C_{\varphi}| - g \geq k'-t.$$
\end{proof}

\begin{lemma} \label{lemma:LB3-beta1}
Let $W = \{x_1, \dots, x_z\}$ be a set of indices of characters in $\alpha^{(1)},$ where
\begin{itemize}
\item for each $j=1, \dots, z,\, hop(x_j) \geq 2,$ 
\item  for $1 \leq j < j' \leq z$ there exist 
$i_j \neq i_{j'}$ such that
$x_j$ is in $Y_{i_j},$ $x_{j'}$ is in $Y_{i_{j'}}.$ 
\end{itemize}

For each $x_j \in W$ let us denote by $x_j^{\gamma}$ the corresponding position in $\alpha^{(1)}_{\gamma}.$
Moreover, we denote by $x_{z+1}^{\gamma}$ the unique position in $\beta^{(1)}$ of the character $\#_{\gamma}^{(1)}$

If there exists $j \in \{1, \dots, z\}$ such that for every $\gamma = 1, \dots, \ell,$ there is no phrase of $\varphi$ ending between 
$x_j^{\gamma}$ and $x_{j+1}^{\gamma}$ then in each one of the substrings $\alpha^{(1)}_{\gamma}$ there is a position $x$ with $hop(x) \geq 3.$
\end{lemma}
\begin{proof}
We begin by observing that for each $\gamma = 1, \dots, \ell,$ there is a phrase ending with $\#^{(1)}_{\gamma}$ since there is no previous occurrence in $\s$ of such a character. Hence, for any $\gamma \neq \gamma'$ there is no phrase of $\varphi$ that
includes characters from both $\alpha^{(1)}_{\gamma}$ and 
$\alpha^{(1)}_{\gamma'}.$

Let us now assume that there is an integer $j$ satisfying the hypothesis. Then, for each $\gamma = 1, \dots,  \ell,$
we have that the phrase $\phi_{\gamma}$ including $x_{j+1}^{\gamma}$ also includes $x_{j}^{\gamma}$ and then also the  character $\#^e_{i_j}$ that appears at the end of the copy of $Y_{i_j}$ in $\alpha^{(1)}_{\gamma}.$

Let us now focus on $\gamma  = 1.$ Since, 
the only occurrence of $\#^e_{i_j}$ preceding $\phi_1$ is
between position $x_j$ and $x_{j+1}$ we have that  
$\phi_1$ must refer to a substring including $x_j$ and $x_{j+1},$ whence $hop(x_{j}^{1}) = hop(x_j) + 1 \geq 3.$ 
Analogously, for each $\gamma = 2, \dots, \ell,$
we have that 
$\phi_{\gamma}$ must refer to a substring including $x_j$ and $x_{j+1},$ or $x_j^{\gamma'}$ and $x_{j+1}^{\gamma'},$
for some $\gamma' < \gamma.$
Therefore, $hop(x_{j}^{\gamma}) \geq 
\min \left\{hop(x_j), hop(x_j^1), \dots, hop(x_j^{\gamma-1}) \right\} + 1 \geq 3.$
\end{proof}

\begin{corollary} \label{corollary:LB3}
If $G$ has no vertex cover of size $k$ and $\varphi$ uses $\leq 4n+6m+k$ phrases in $\alpha^{(1)}$ and $< \ell+ z $   phrases in $\beta^{(1)}$  then on 
$\alpha^{(2)}$ it produces $\ell$ positions $x^{(1)}_1, \dots, x^{(1)}_{\ell}$ within $\beta^{(1)}$ with $hop \geq 3$ such that $\#^{(1)}_j$ is between position $x^{(1)}_j$ and position $x^{(1)}_{j+1}$
for each $j=1, \dots, \ell-1.$ 
\end{corollary}
\begin{proof}
If $\tau(G) =  k' > k$ and $|\varphi(\alpha^{(1)})| = 4n+6m+t \leq 4n+6m+k,$ by Lemma \ref{lemma:LB2-alpha1}, there are $z \geq k' - t$ positions $x_1, \dots, x_z,$ with each $x_j$ in a distinct $Y_{i_j}$ and such that $hop(x_j) = 2.$      
Since $\varphi$ uses $< \ell+ g $ phrases in $\beta^{(1)}$
and it must have a phrase ending at the character
$\#^{(1)}_{\gamma}$ for each $\gamma = 1, \dots, \ell,$
it follows that, using the notation of Lemma \ref{lemma:LB3-beta1}, there is a $j \in \{1, \dots, z+1\}$ such that there is no phrase ending between $x_j^{\gamma}$ and $x_{j+1}^{\gamma+`1}.$
Then, the desired result follows by Lemma \ref{lemma:LB3-beta1}.
\end{proof}

\begin{lemma} \label{lemma:LB2l}
Let $\varphi$ be a parsing of $\s$ such that for some $i \in \{1,2, \dots, c-2\},$ and for each $j = 1, 2, \dots, \ell$ in the substring $\alpha^{(i)}_j$ of
$\beta^{(i)}$ there is a position $x_j$ such that $hop(x_j) > i+1.$
For each $\gamma = 1, \dots, \ell$ let $x^{\gamma}_j$ be the position in $\alpha^{(i+1)}_\gamma$ corresponding to $x_j$. 
If $\varphi$ uses strictly less than $2\ell$ phrases in $\beta^{(i+1)}$ then 
there is a $j \in \{1, 2, \dots, \ell\}$ such that 
for each $\gamma=1, \dots, \ell$ the position $x_{j}^{\gamma}$ in $\alpha^{(i+1)}_{\gamma}$ satisfies $hop(x_{j}^{\gamma}) > i+2.$
\end{lemma}
\begin{proof}
We begin by observing that for each $\gamma = 1, \dots, \ell,$ there is a phrase ending with $\#^{(i+1)}_{\gamma}$ since there is no previous occurrence in $\s$ of such a character. Hence, for any $\gamma \neq \gamma'$ there is no phrase of $\varphi$ that
includes characters from both $\alpha^{(i+1)}_{\gamma}$ and 
$\alpha^{(i+1)}_{\gamma'}.$

If besides these $\ell$ phrases, $\varphi$ uses $< \ell$ additional phrases in $\beta^{(i)}$ it follows that  
there is an integer 
$j \in \{1, 2, \dots, \ell\}$ such that
for each $\gamma=1, \dots, \ell$ no phrase of $\varphi$ ends between $x^{\gamma}_j$ and $x^{\gamma}_{j+1}$

Then, for each $\gamma = 1, \dots,  \ell,$
we have that the phrase $\phi_{\gamma}$ including $x_{j+1}^{\gamma}$ also includes $x_{j}^{\gamma}$ and then also the  character $\#^{(i)}_{j}$ that appears at the end of the $j$th copy of $\alpha^{(i)}$ in $\alpha^{(i+1)}_{\gamma}.$

Let us now focus on $\gamma  = 1.$ Since, 
the only occurrence of $\#^{(i)}_{j}$ preceding $\phi_1$ is
between position $x_j$ and $x_{j+1}$ we have that  
$\phi_1$ must refer to a substring including $x_j$ and $x_{j+1},$ whence $hop(x_{j}^{1}) = hop(x_j) + 1 \geq i+2.$ 
Analogously, for each $\gamma = 2, \dots, \ell,$
we have that 
$\phi_{\gamma}$ must refer to a substring including $x_j$ and $x_{j+1},$ or $x_j^{\gamma'}$ and $x_{j+1}^{\gamma'},$
for some $\gamma' < \gamma.$
Therefore, $hop(x_{j}^{\gamma}) \geq 
\min \left\{hop(x_j), hop(x_j^1), \dots, hop(x_j^{\gamma-1}) \right\} + 1 \geq i+2.$
\end{proof}

\begin{lemma} \label{lemma:noVC-noParsing} 
If $\tau(G) > k$ then any $c$-BLZ parsing $\varphi$ of $\s$ has size $> 4n+6m+k+ (c-1)\ell.$ 
\end{lemma}
\begin{proof}
Let  $\varphi$ be a $c$-BLZ parsing of $\s.$

\noindent
{\em Claim 1.} for each $i=1, \dots, c-1, \, |\varphi(\beta^{(i)}| \geq \ell.$ The claim easily follows from the fact that in $\beta^{(i)}$ there is the first occurrence of the characters $\#^{(i)}_1, \dots, \#^{(i)}_{\ell},$ hence each one of these characters must be the end of a distinct phrase.

\medskip

We now argue by cases.

\noindent
{\em Case 1.} $|\varphi(\alpha^{(1)})| > 4n+ 6m + k.$
Then, by Claim 1, we have that: 
$$|\varphi(\s)| = |\varphi(\alpha^{(1)})| + \sum_{i=1}^{c-1} |\varphi(\beta^{(i)})| > 4n+ 6m + k + \ell(c-1).$$

\medskip

\noindent
{\em Case 2} $|\varphi(\alpha^{(1)})| = 4n+ 6m + t,$ for some $t\leq k.$

Since $\tau(G) > k,$ by Lemma \ref{lemma:LB2-alpha1} we have that 
 for some integer $z \geq  \tau(G) - t$ and for each $i=1, \dots , z$ there is a positions $x_i$ in the substring $Y_i$ of $\alpha^{(1)}$ such that $hop(x_i) \geq 2.$  

\medskip

{\em Subcase 2.1} $|\varphi(\beta^{(1)})| \geq \ell+z.$

Then, again by Claim1, we have that in total $|\varphi|$ satisfies the desired bound:
\begin{multline}
|\varphi(\s)| = |\varphi(\alpha^{(1)})| + |\varphi(\beta^{(1)})| \sum_{i=2}^{c-1} |\varphi(\beta^{(i)})| \geq 4n+6m+t + \ell+ z + \ell(c-2) \\ \geq  4n+6m + t + \ell(c-1) + \tau(G) - t > 
4n+ 6m + k + \ell(c-1).
\end{multline}

{\em Subcase 2.2} $|\varphi(\beta^{(1)})| < \ell+z.$

Then, by Corollary \ref{corollary:LB3} ,
$\varphi$ uses in $\beta^{(2)}$ are less than $z$ it means that 
it follows, 
that for each $\gamma= 1, \dots, \ell$ in the substring $\alpha^{(2)}_{\gamma}$ there is a position  $x^{(2)}_{\gamma}$ 
such that $hop(x^{(2)}_{\gamma}) \geq 3.$ 

We now observe that under the standing hypothesis there must exist an integer $j^* \in \{2, 3, \dots, c-1\}$ such that 
$|\varphi(\beta^{(j^*)})| \geq 2\ell.$ For otherwise, 
by repeated application of Lemma \ref{lemma:LB2l} we have that
for $i=2, \dots, c-1$ 
 there  are $\ell$ position with hop-number $\geq i+2.$ This in 
 particular means that in $\beta^{(c-1)}$ there are positions with 
 hop-number $c+1$ contradicting the hypothesis that $\varphi$ is a
 $c$-BLZ parsing.

 Therefore, we can again bound the number of phrases of $\varphi$ as follows:
 \begin{multline}
 |\varphi(\s)| = |\varphi(\alpha^{(1)})| + |\varphi(\beta^{(j^*)})| +  \sum_{\substack{i=1 \\ i \neq j^* }}^{c-1} |\varphi(\beta^{(i)})| \geq 4n+6m+t + 2\ell+  \ell(c-2) \\ \geq  4n+6m + \ell + \ell(c-1) > 
4n+ 6m + k + \ell(c-1),
 \end{multline}
where we have used the fact that in each $\beta^{(i)}$ any parsing must use at least $\ell$ phrases and the assumption $\ell > k.$

We have shown that in all possible cases the size of $\varphi$ satisfied the desired bound.
\end{proof}

\bigskip

\noindent
{\bf The proof of Theorem \ref{thm:minVC-minParsing}.}
The result  is an immediate corollary of Lemma \ref{lemma:VC2Parsing} and Lemma \ref{lemma:noVC-noParsing}. 
\qed
\end{document}